\documentclass[a4paper, USenglish, cleveref, autoref, thm-restate,pagebackref]{lipics-v2019}

\title{Parameterized Algorithms for Matrix Completion With Radius Constraints} 

\author{Tomohiro Koana}{Technische Universit\"at Berlin, Faculty~IV, Algorithmics and Computational Complexity}{koana@campus.tu-belin.de}{https://orcid.org/}{Partially supported by the DFG project MATE (NI 369/17).}

\author{Vincent Froese}{Technische Universit\"at Berlin, Faculty~IV, Algorithmics and Computational Complexity}{vincent.froese@tu-berlin.de}{https://orcid.org/}{}

\author{Rolf Niedermeier}{Technische Universit\"at Berlin, Faculty~IV, Algorithmics and Computational Complexity}{rolf.niedermeier@tu-berlin.de}{https://orcid.org/}{}

\authorrunning{T.~Koana, V.~Froese, and R.~Niedermeier} 

\Copyright{Tomohiro Koana, Vincent Froese, and Rolf Niedermeier} 



\acknowledgements{We are grateful to Christian Komusiewicz for helpful feedback on an earlier version of this work and to Stefan Szeider for pointing us to reference~\cite{EGKOS19}.}

\nolinenumbers 

\hideLIPIcs  

\EventEditors{}
\EventNoEds{0}
\EventLongTitle{}
\EventShortTitle{}
\EventAcronym{}
\EventYear{}
\EventDate{}
\EventLocation{}
\EventLogo{}
\SeriesVolume{}
\ArticleNo{}
\usepackage{amsmath}

\usepackage{tabularx,booktabs,multirow}
\newcolumntype{m}{>{\hsize=.5\hsize}X}

\usepackage{tikz}
\usetikzlibrary{arrows,decorations.pathmorphing,decorations.pathreplacing,backgrounds,positioning,fit,matrix}
\usetikzlibrary{shapes,calc}

\usepackage[ruled]{algorithm}
\usepackage[noend]{algpseudocode}
\algnewcommand\algorithmicinput{\textbf{Input:}}
\algnewcommand\algorithmicoutput{\textbf{Task:}}
\algnewcommand\Input{\item[\algorithmicinput]}
\algnewcommand\Output{\item[\algorithmicoutput]}

\usepackage{todonotes}

\newcommand{\NN}{\mathbb{N}} 
\newcommand{\RR}{\mathbb{R}} 

\newcommand{\problemdef}[4]{
\begin{center}
  \begin{minipage}{0.95\textwidth}
    \normalsize\textsc{#2} \smallskip \\
    \begin{tabularx}{\textwidth}{@{}l@{\hspace{3pt}}X}
      \normalsize\textbf{Input:} & \normalsize#3 \\
      \normalsize\textbf{#1:}    & \normalsize#4
    \end{tabularx}
  \end{minipage}
\end{center}
}
\newcommand{\dprob}[4][Question]{\problemdef{#1}{#2}{#3}{#4}}

\begin{document}

\maketitle

\begin{abstract}
Considering matrices with missing entries, we study NP-hard matrix completion 
problems where the resulting completed matrix shall have limited (local) radius.
In the pure radius version, this means that the goal is to fill in the entries
such that there exists a `center string' which 
has Hamming distance to all matrix
rows as small as possible. In stringology, this problem is also 
known as \textsc{Closest String with Wildcards}. In the local radius version,
the requested center string must be one of the rows of the completed
matrix.

Hermelin and Rozenberg [CPM~2014, TCS~2016] performed parameterized
complexity studies for \textsc{Closest String with Wildcards}.
We answer one of their open questions, fix a bug concerning a fixed-parameter 
tractability result in their work, and improve some upper running time bounds.
For the local radius case, we 
reveal a computational complexity dichotomy. In general, our results indicate that,
although being NP-hard as well, this variant often allows for faster 
(fixed-parameter) algorithms.
\end{abstract}

\section{Introduction}

In many applications data can only be partially measured, which leads to incomplete data records with missing entries. A common problem in data mining, machine learning, and computational biology 
is to infer these missing entries.
In this context, matrix completion problems
play a central role. Here the goal is to fill in the unknown entries of an incomplete data matrix such that certain measures regarding the completed matrix are optimized. Ganian et al.~\cite{GKOS18} recently studied parameterized 
algorithms for two variants of matrix completion problems; their goal was 
to either minimize the rank or to minimize the number of distinct rows in the 
completed matrix. In this work, we focus our study on another variant, 
namely, minimizing the `radius' of the completed matrix.
Indeed, this is closely related to the topic of consensus (string)
problems, which received a lot of attention in stringology and 
particularly with respect to parameterized complexity studies~\cite{BHKN14}.
Indeed, radius minimization for 
incomplete matrices is known in the stringology community under
the name
\textsc{Closest String with Wildcards}~\cite{HR15}, a generalization of 
the frequently studied \textsc{Closest String} problem.
Herein, an incomplete matrix shall be completed such that there exists a vector that is not too far from each row vector of the completed matrix in terms of the Hamming distance (that is, the completed matrix has small \emph{radius}).
We consider this radius minimization problem and also a local variant where all row vectors of the completed matrix must be close to a row vector 
in the matrix.

\begin{figure}[t]
  \centering
  \begin{tikzpicture}[
    scale=0.5,
    every node/.style={font=\footnotesize}
  ]
    \node at (0.5, 0.5) {$+$};
    \node at (0.5, 1.5) {$-$};
    \node at (0.5, 3.5) {$-$};
    \node at (1.5, 1.5) {$*$};
    \node at (0.5, 2.5) {$*$};

    \node at (1.5, 0.5) {$\beta$};
    \node at (1.5, 2.5) {$\gamma$};
    \node at (1.5, 3.5) {$\alpha$};

    \node at (2.5, 0.5) {4.2};
    \node at (2.5, 1.5) {7.3};
    \node at (2.5, 2.5) {4.2};
    \node at (2.5, 3.5) {$*$};

    \node at (3.5, 0.5) {$>$};
    \node at (3.5, 3.5) {$>$};
    \node at (3.5, 1.5) {$*$};
    \node at (3.5, 2.5) {$*$};

    \node at (4.5, 0.5) {$0$};
    \node at (4.5, 1.5) {$1$};
    \node at (4.5, 2.5) {$2$};
    \node at (4.5, 3.5) {$1$};
    \draw[help lines] (0,0) grid (5, 4);
  \begin{scope}[shift={(7, 0.75)}]
    \fill[lipicsLightGray] (0, 0) rectangle (1, 1);
    \fill[lipicsLightGray] (1, 3) rectangle (2, 4);
    \fill[lipicsLightGray] (1, 2) rectangle (2, 3);
    \fill[lipicsLightGray] (2, 1) rectangle (3, 2);
    \fill[lipicsLightGray] (4, 0) rectangle (5, 1);
    \fill[lipicsLightGray] (4, 2) rectangle (5, 3);
    \node at (0.5, 0.5) {$+$};
    \node at (0.5, 1.5) {$-$};
    \node at (0.5, 2.5) {$-$};
    \node at (0.5, 3.5) {$-$};

    \node at (1.5, 0.5) {$\beta$};
    \node at (1.5, 1.5) {$\beta$};
    \node at (1.5, 2.5) {$\gamma$};
    \node at (1.5, 3.5) {$\alpha$};

    \node at (2.5, 0.5) {4.2};
    \node at (2.5, 1.5) {7.3};
    \node at (2.5, 2.5) {4.2};
    \node at (2.5, 3.5) {4.2};

    \node at (3.5, 0.5) {$>$};
    \node at (3.5, 1.5) {$>$};
    \node at (3.5, 2.5) {$>$};
    \node at (3.5, 3.5) {$>$};

    \node at (4.5, 0.5) {$0$};
    \node at (4.5, 1.5) {$1$};
    \node at (4.5, 2.5) {$2$};
    \node at (4.5, 3.5) {$1$};
    \draw[help lines] (0,0) grid (5, 4);

    \begin{scope}[shift={(0,-1.5)}]
      \node at (0.5, 0.5) {$-$}; \node at (1.5, 0.5) {$\beta$}; \node at (2.5, 0.5) {4.2}; \node at (3.5, 0.5) {$>$}; \node at (4.5, 0.5) {$1$};
      \draw[help lines] (0,0) grid (5, 1);
    \end{scope}
  \end{scope}
  \begin{scope}[shift={(14, 0.75)}]
    \fill[lipicsLightGray] (0, 0) rectangle (1, 1);
    \fill[lipicsLightGray] (1, 0) rectangle (2, 1);
    \fill[lipicsLightGray] (1, 3) rectangle (2, 4);
    \fill[lipicsLightGray] (2, 1) rectangle (3, 2);
    \fill[lipicsLightGray] (4, 0) rectangle (5, 1);
    \fill[lipicsLightGray] (4, 1) rectangle (5, 2);
    \fill[lipicsLightGray] (4, 3) rectangle (5, 4);
    \node at (0.5, 0.5) {$+$};
    \node at (0.5, 1.5) {$-$};
    \node at (0.5, 2.5) {$-$};
    \node at (0.5, 3.5) {$-$};

    \node at (1.5, 0.5) {$\beta$};
    \node at (1.5, 1.5) {$\gamma$};
    \node at (1.5, 2.5) {$\gamma$};
    \node at (1.5, 3.5) {$\alpha$};

    \node at (2.5, 0.5) {4.2};
    \node at (2.5, 1.5) {7.3};
    \node at (2.5, 2.5) {4.2};
    \node at (2.5, 3.5) {4.2};

    \node at (3.5, 0.5) {$>$};
    \node at (3.5, 1.5) {$>$};
    \node at (3.5, 2.5) {$>$};
    \node at (3.5, 3.5) {$>$};

    \node at (4.5, 0.5) {$0$};
    \node at (4.5, 1.5) {$1$};
    \node at (4.5, 2.5) {$2$};
    \node at (4.5, 3.5) {$1$};
    \draw[help lines] (0,0) grid (5, 4);

    \begin{scope}[shift={(0,-1.5)}]
        \draw[help lines] (0,0) grid (5, 1);
      \node at (0.5, 0.5) {$-$}; \node at (1.5, 0.5) {$\gamma$}; \node at (2.5, 0.5) {4.2}; \node at (3.5, 0.5) {$>$}; \node at (4.5, 0.5) {$2$};
    \end{scope}
  \end{scope}
  \end{tikzpicture}
  \caption{An example of \textsc{MinRMC} and \textsc{MinLRMC}.
  The incomplete input matrix is depicted in the left (a $*$ denotes a missing entry).
  Optimal solutions for \textsc{MinRMC} ($d = 2$) and \textsc{MinLRMC} ($d = 3$) are shown in the middle and right.
  Entries differing from the solution vector are marked gray.
  }
  \label{fig:mcp:example}
\end{figure}

Given the close relation to \textsc{Closest String}, 
which is NP-hard already for binary strings~\cite{FL97}, all problems we 
study in this work are NP-hard in general. However, we provide 
several positive algorithmic results, namely fixed-parameter tractability with respect to natural parameters or even polynomial-time solvability for special cases.
Formally, we study the following problems (see \Cref{fig:mcp:example} for illustrative examples):

\dprob{Minimum Radius Matrix Completion (MinRMC)}
{An incomplete matrix $\mathbf{S} \in (\Sigma \cup \{ * \})^{n \times \ell}$ and $d \in \NN$.}
{Is there a `completion' $\mathbf{T} \in \Sigma^{n \times \ell}$ of $\mathbf{S}$ such that $\delta(v, \mathbf{T}) \le d$ for some vector $v \in \Sigma^\ell$?}

\dprob{Minimum Local Radius Matrix Completion (MinLRMC)}
{An incomplete matrix $\mathbf{S} \in (\Sigma \cup \{ * \})^{n \times \ell}$ and $d \in \NN$.}
{Is there a `completion' $\mathbf{T} \in \Sigma^{n \times \ell}$ of $\mathbf{S}$ such that $\delta(v, \mathbf{T}) \le d$ for some vector $v \in \mathbf{T}$?}

Here, missing entries are denoted by~$*$ and $\delta$ denotes the Hamming distance.
In fact, our results for \textsc{MinRMC} 
also hold for the following more general problem:
\dprob{Constraint Radius Matrix Completion (ConRMC)}
{An incomplete matrix $\mathbf{S} \in (\Sigma \cup \{ * \})^{n \times \ell}$ and $d_1, \dots, d_n \in \NN$.}
{Is there a row vector $v \in \Sigma^\ell$ such that $\delta(v, \mathbf{S}[i]) \le d_i$ for each $i \in [n]$?}

\paragraph*{Related work.}
Our most important reference point is the work of Hermelin and 
Rozenberg~\cite{HR15} who analyzed the parameterized complexity of 
\textsc{MinRMC} (under the name \textsc{Closest String with 
Wildcards})
with respect to several problem-specific 
parameters (also see \Cref{table:rmcp}).
In particular, they provided fixed-parameter tractability results 
for the parameters 
number~$n$ of rows,  number~$\ell$ of columns, and radius~$d$ 
combined with maximum number~$k$ of missing entries per row.
However, we will show that their fixed-parameter tractability 
result for the combined parameter~$k+d$ is flawed.
Moreover, they
showed a computational complexity dichotomy for binary inputs 
between radius~1 (polynomial time) and radius~2 (NP-hard) (see 
\Cref{table:rmcp} for a complete overview). 

As mentioned before, Ganian et al.~\cite{GKOS18} started research on 
the parameterized complexity of two related matrix completion problems
(minimizing rank and minimizing number of distinct rows). 
Very recently, Eiben et al.~\cite{EGKOS19} studied generalizations of our 
problems by demanding that the completed matrix is clustered 
into several submatrices of small (local) radius---basically, our work
studies the case of a single cluster.
They proved fixed-parameter tractability (problem kernels of superexponential size) with respect to the combined parameter~$(c,d,r)$.
Here, $c$~is the number of clusters and $r$~is the minimum number of rows and columns covering all missing entries.
They also proved that dropping any of $c$, $d$, or~$r$ results in parameterized intractability even for binary alphabet.
Amir et al.~\cite{AFRS14} showed that the clustering version of \textsc{MinLRMC} on complete matrices with unbounded alphabet size is NP-hard when restricted to only two columns.
Note that fixed-parameter tractability for the clustering variant implies fixed-parameter tractability for \textsc{MinRMC} and \textsc{MinLRMC} with respect 
to~$d + r$.
Indeed, to reach for (more) practical algorithms, we consider an alternative parameterization by the maximum number~$k$ of missing entries in any row vector (which can be much smaller than~$r$).

\paragraph*{Our contributions.}
We survey our and previous results in \Cref{table:rmcp}. Notably, all of our results for \textsc{MinRMC} indeed also hold for 
the more general \textsc{ConRMC} when setting $d:=\max\{ d_1, d_2, \ldots, d_n\}$.

\begin{table}[t]
  \caption{Overview of known results and our new results for \textsc{MinRMC} and \textsc{MinLRMC}.
  Notation: $n$---number of rows, $\ell$---number of columns, $|\Sigma|$---alphabet size, $d$---distance bound, $k$---maximum number of missing entries in any row vector. Note that all our results for \textsc{MinRMC} also hold for \textsc{ConRMC}.}
  \label{table:rmcp}
  \begin{tabularx}{\textwidth}{l@{\hspace{10pt}}XmXm@{\hspace{0pt}}}
    \toprule
    Parameter & \textsc{MinRMC}          & Reference                              & \textsc{MinLRMC} & Reference \\
    \midrule
    $n$       & $O^*(2^{2^{O(n \log n)}})$       & \cite{HR15}            & $O^*(2^{O(n^2 \log n)})$ &\cref{lemma:minlrmc:fptn} \\
              & $O^*(2^{O(n^2 \log n)})$         & \cite{KKM17}           &                          & \\ [0.7ex]
    $\ell$    & $O^*(2^{\ell^2 / 2})$            & \cite{HR15}            & $O^*(\ell^{\ell})$       & \cref{thm:minlrmc:fptk} \\
              & $O^*(\ell^{\ell})$               & \cref{thm:rmcpfptwrtl} &  & \\ [0.5ex]
    $d = 1$   & $O(n \ell^2)$ for $|\Sigma| = 2$ & \cite{HR15}            & $O(n^2 \ell)$ & \cref{thm:minlrmc:d1} \\
              & $O(n \ell)$                      & \cref{thm:rmcpd1poly}  &  &\\
    $d = 2$   & NP-hard for $|\Sigma| = 2$       & \cite{HR15}            & NP-hard for $|\Sigma| = 2$ & \cite{HR15} \\ [0.7ex]
    $k$       & NP-hard for $k = 0$              & \cite{FL97}            & $O^*(k^k)$ & \cref{thm:minlrmc:fptk} \\
    $d + k$   & $O^*((d + 1)^{d + k})$           & \cref{theorem:fpt}     & & \\
    $d + k + |\Sigma|$ & $O^*(|\Sigma|^k \cdot d^d)$ & \cite{HR15}        & $O^*(|\Sigma|^k)$ & trivial \\ 
                       & $O^*(2^{4d + k} \cdot |\Sigma|^{d + k})$ & \cref{theorem:fpt2} & \\
    \bottomrule
  \end{tabularx}
\end{table}

Let us highlight a few of our new results in comparison with previous work.
For \textsc{MinRMC}, we give a linear-time algorithm for the case radius $d = 1$ 
and arbitrary alphabet.
This answers an open question of Hermelin and Rozenberg \cite{HR15}.
We also show fixed-parameter tractability with respect to the combined 
parameter~$d+k$, which was already claimed by Hermelin and 
Rozenberg~\cite{HR15} but was flawed, as we will point out 
by providing a counter-example to their algorithm.
Lastly, inspired by known results for \textsc{Closest String}, 
we give a more efficient algorithm for small alphabet size.

As regards \textsc{MinLRMC}, we show that it can be solved in polynomial time when $d = 1$. This yields a computational complexity dichotomy since 
\textsc{MinLRMC} is NP-hard for~$d=2$.
Moreover, we show that \textsc{MinLRMC} is fixed-parameter tractable with 
respect to the maximum number~$k$ of missing entries per row.
Remarkably, this stands in sharp contrast to \textsc{MinRMC}, 
which is NP-hard even for~$k=0$.

\section{Preliminaries}

For $m \le n \in \NN$, let $[m, n]:=\{ m, \dots, n \}$ and $[n]:=[1, n]$.

Let $\mathbf{T} \in \Sigma^{n \times \ell}$ be an $(n \times \ell)$-matrix over a finite alphabet $\Sigma$.
Let $i \in [n]$ and $j \in [\ell]$.
We use $\mathbf{T}[i, j]$ to denote the character in the~$i$-th row and~$j$-th column of~$\mathbf{T}$.
We use $\mathbf{T}[i, :]$ (or $\mathbf{T}[i]$ in short) to denote the \emph{row vector} $(\mathbf{T}[i, 1], \dots, \mathbf{T}[i, \ell])$ and $\mathbf{T}[:, j]$ to denote the \emph{column vector} $(\mathbf{T}[1, j], \dots, \mathbf{T}[n, j])^T$.
For any subsets $I \subseteq [n]$ and $J \subseteq [\ell]$, we write $\mathbf{T}[I, J]$ to denote the submatrix obtained by omitting rows in $[n] \setminus I$ and columns in $[\ell] \setminus J$ from $\mathbf{T}$.
We abbreviate $\mathbf{T}[I, [\ell]]$ and $\mathbf{T}[[n], J]$ as $\mathbf{T}[I, :]$ (or $\mathbf{T}[I]$ for short) and $\mathbf{T}[:, J]$, respectively.
We use the special character $*$ for a \emph{missing} entry.
A matrix $\mathbf{S} \in (\Sigma \cup \{ * \})^{n \times \ell}$ that contains a missing entry is called \emph{incomplete}.
We say that $\mathbf{T} \in \Sigma^{n \times \ell}$ is a \emph{completion} of $\mathbf{S} \in (\Sigma \cup \{ * \})^{n \times \ell}$ if either $\mathbf{S}[i, j] = *$ or $\mathbf{S}[i, j] = \mathbf{T}[i, j]$ holds for all $i \in [n]$ and $j \in [\ell]$.

Let $v, v' \in ( \Sigma \cup \{ * \})^\ell$ be row vectors and let $\sigma \in \Sigma \cup \{ * \}$.
We write $P_\sigma(v)$ to denote the set $\{ j \in [\ell] \mid v[j] = \sigma \}$ of column indices where the corresponding entries of~$v$ are~$\sigma$.
We write~$Q(v, v')$ to denote the set $\{ j \in [\ell] \mid v[j] \ne v'[j] \wedge v[j]\ne * \wedge v'[j]\ne * \}$ of column indices where~$v$ and~$v'$ disagree (not considering positions with missing entries).
The \emph{Hamming distance} between $v$ and $v'$ is $\delta(v, v') := |Q(v, v')|$.
For $\mathbf{S} \in (\Sigma \cup \{ * \})^{n \times \ell}$ and $v \in (\Sigma \cup \{ * \})^\ell$,
let $\delta(v, \mathbf{S}):=\max_{i \in [n]} \delta(v, \mathbf{S}[i])$.
The binary operation $v \oplus v'$ replaces the missing entries of~$v$ with the character in $v'$ in the corresponding position, given that $v'$ contains no missing entry.
We sometimes use string notation $\sigma_1 \sigma_2 \sigma_3$ to represent the row vector~$(\sigma_1, \sigma_2, \sigma_3)$.

\paragraph*{Parameterized Complexity.}

We sometimes use the $O^*$-notation which suppresses polynomial factors in the input size.
A \emph{parameterized problem} $\Pi$ is a set of instances $(I,k)\in \Sigma^* \times \NN$, where $k$ is called the \emph{parameter} of the instance.
A parameterized problem is \emph{fixed-parameter tractable} if $(I, k) \in \Pi$ can be determined in $f(k) \cdot |I|^{O(1)}$ time for an arbitrary computable function~$f$.
An algorithm with such a running time is called a \emph{fixed-parameter algorithm}.

\section{Linear-time algorithm for radius \texorpdfstring{$d = 1$}{d = 1}}
\label{sec:radmin:d}

Hermelin and Rozenberg \cite[Theorem 6]{HR15} gave a reduction from \textsc{MinRMC} to \textsc{2-SAT} for the case $|\Sigma| = 2$ and $d = 1$, resulting in an $O(n \ell^2)$-time algorithm.
We provide a more efficient reduction to \textsc{2-SAT}, exploiting the compact encoding $C_{\le 1}$ of the ``at-most-one'' constraint by Sinz~\cite{sinz2005towards}.
Let $L = \{ l_1, \dots, l_m \}$ be a set of $m$ literals.
The encoding uses $m - 1$ additional variables $r_1, \dots, r_{m - 1}$ and it is defined as follows:
\begin{align*}
  C_{\le 1}(L)
  &= (\neg l_1 \vee r_1) \wedge (\neg l_m \vee \neg r_{m - 1}) \\
  &\wedge \bigwedge_{2 \le j \le m - 1} \left((\neg l_{j} \vee \neg r_{j - 1}) \wedge (\neg l_{j} \vee r_j) \wedge (\neg r_{j - 1} \vee r_j) \right).
\end{align*}
Note that if $l_j$ is true for some $j \in [m]$, then $r_{j}, \dots, r_{m - 1}$ are all true and $r_{1}, \dots, r_{j - 1}$ are all false.
Hence, at most one literal in $L$ can be true.

Actually, our algorithm solves \textsc{ConRMC}, the generalization of \textsc{MinRMC} where the distance bound can be specified for each row vector individually.

\begin{theorem}
  \label{thm:rmcpd1poly}
If $\max_{i \in [n]} d_i = 1$, then \textsc{ConRMC} can be solved in $O(n \ell)$~time. 
\end{theorem}
\begin{proof}
  We reduce \textsc{ConRMC} to \textsc{2-SAT}.
  Let $I_d := \{ i \in [n] \mid d_i = d \}$ be the row indices for which the distance bound is $d$ for $d \in \{ 0, 1 \}$.
  We define a variable $x_{j, \sigma}$ for each $j \in [\ell]$ and $\sigma \in \Sigma$.
  The intuition behind our reduction is that the $j$-th entry of the solution vector $v$ becomes $\sigma$ when $x_{j, \sigma}$ is true.
  We give the construction of a 2-CNF formula $\phi$ in three parts $\phi_1, \phi_2, \phi_3$ (that is, $\phi = \phi_1 \wedge \phi_2 \wedge \phi_3$).
  \begin{itemize}
    \item
      Let $X_j = \{ x_{j, \sigma} \mid \sigma \in \Sigma \}$ for each $j \in [\ell]$.
      The first subformula will ensure that at most one character is assigned to each entry of the solution vector $v$:
      \begin{align*}
        \phi_1 &= \bigwedge_{j \in [\ell]} C_{\le 1} (X_j).
      \end{align*}
    \item
      Subformula $\phi_2$ handles distance-0 constraints:
      \begin{align*}
        \phi_2 &= \bigwedge_{i \in I_0} \bigwedge_{\substack{j \in [\ell] \\ \mathbf{S}[i, j] \ne *}} (x_{j, \mathbf{S}[i, j]}).
      \end{align*}
    \item
      Finally, subformula~$\phi_3$ guarantees that the solution vector $v$ deviates from each row vector of $\mathbf{S}[I_1]$ in at most one position.
      \begin{align*}
        \phi_3 &= \bigwedge_{i \in I_1} C_{\le 1} (\{ \neg x_{j, \mathbf{S}[i, j]} \mid j \in [\ell], \mathbf{S}[i, j] \ne * \}).
      \end{align*}
  \end{itemize}
  Note that our construction uses $O(|\Sigma| \cdot \ell)$ variables and $O((n + |\Sigma|) \cdot \ell) = O(n \ell)$ clauses.
  We prove the correctness of the reduction.

  $(\Rightarrow)$
  Suppose that there exists a vector $v \in \Sigma^\ell$ such that $\delta(v, \mathbf{S}[i]) \le d_i$ holds for each $i \in [n]$.
  For each $j \in [\ell]$ and $\sigma \in \Sigma$, we set $x_{j, \sigma}$ to true if $v[j] = \sigma$, and false otherwise.
  It is easy to see that this truth assignment satisfies $\phi$.

  $(\Leftarrow)$
  Suppose that there exists a satisfying truth assignment $\varphi$.
  Let $J^*$ denote the column indices $j \in [\ell]$ such that $\varphi(x_{j, \sigma}) = 0$ for all $\sigma \in \Sigma$.
  Note that at most one variable in~$X_j$ is set to true in $\varphi$ for each $j \in [\ell]$.
  It follows that, for each $j \in [\ell] \setminus J^*$, there exists exactly one character $\sigma_j \in \Sigma$ satisfying $\varphi(x_{j, \sigma_j}) = 1$ and we assign $v[j] = \sigma_j$.
  For each $j \in J^*$, we set $v[j] = \sigma^*$ for some arbitrary character $\sigma^* \in \Sigma$.
  The formula $\phi_2$ ensures that $\delta(v, \mathbf{S}[i]) = 0$ holds for each $i \in I_0$.
  Moreover, $\phi_3$ ensures that there is at most one column index $j \in [\ell]$ such that $\mathbf{S}[i, j] \ne *$ and $\mathbf{S}[i, j] \ne v[j]$ for each $i \in I_1$.
\end{proof}

Note that \textsc{MinRMC} (and thus \textsc{ConRMC}) is 
NP-hard for $|\Sigma| = 2$ and $d = 2$~\cite{HR15}.
Thus, our result implies a complete complexity dichotomy regarding~$d$.
We remark that this dichotomy also holds for \textsc{MinLRMC}
since there is a simple reduction from \textsc{MinLRMC} to \textsc{ConRMC}.
To solve an instance $(\mathbf{S}, d)$ of \textsc{MinLRMC},
we solve~$n$ instances of \text{ConRMC}:
For each $i \in [n]$, we solve the instance $(\mathbf{S}, d_1, \dots, d_n)$ where $d_{i'} = d$ for each $i' \in [n] \setminus \{ i \}$ and $d_i = 0$.
Clearly, $(\mathbf{S}, d)$ is a \textbf{Yes}-instance if and only if at least one \textsc{ConRMC}-instance is a \textbf{Yes}-instance.
This yields the following.

\begin{corollary}
  \label{thm:minlrmc:d1}
  \textsc{MinLRMC} can be solved in $O(n^2 \ell)$ time when $d = 1$.
\end{corollary}

Since \textsc{ConRMC} is solvable in $O^*(2^{O(n^2 \log n)})$ time~\cite{KKM17}, we also obtain the following result, where the running time bound 
only depends on the number~$n$ of rows.
\begin{corollary}
  \label{lemma:minlrmc:fptn}
  \textsc{MinLRMC} can be solved in $O^*(2^{O(n^2 \log n)})$ time.
\end{corollary}

Finally, we remark that \textsc{ConRMC} can be solved in linear time 
for binary alphabet $\Sigma = \{0,1\}$ if 
$d_i \ge \ell - 1$ for all $i \in [n]$ (the problem remains NP-hard in the case of unbounded alphabet size \cite{LMS18} even if $d_i \ge \ell - 1$ for all $i \in [n]$):
First, we remove each row vector with distance bound $\ell$.
We also remove every row vector with at least one missing entry since it has distance at most $\ell - 1$ from any vector of length $\ell$.
We then remove every duplicate row vector.
This can be achieved in linear time: 
We sort the row vectors lexicographically using radix sort and we compare each row vector to the adjacent row vectors in the sorted order.
We return \textbf{Yes} if and only if there are at most $2^\ell - 1$ row vectors, because each distinct row vector $u \in \{ 0, 1 \}^\ell$ excludes exactly one row vector $\overline{u} \in \{ 0, 1 \}^\ell$ where $\overline{u}[j] = 1 - u[j]$ for each $j \in [\ell]$.
Summarizing, we arrive at the following.
\begin{proposition}
	If $\Sigma = \{0,1\}$ and $d_i \ge \ell - 1$ for all $i \in [n]$, then \textsc{ConRMC} can be solved in linear time.
\end{proposition}

\section{Parameter number \texorpdfstring{$\ell$}{l} of columns}
\label{sec:radmin:l}

Hermelin and Rozenberg~\cite[Theorem 3]{HR15} showed that one can solve \textsc{MinRMC} in $O(2^{\ell^2 / 2} \cdot n \ell)$ time using a search tree algorithm.
We use a more refined recursive step to obtain a better running time (see \Cref{algo:wrtl}).
In particular we employ a trick used by Gramm et al.~\cite{gramm2003fixed} in order to reduce the search space to $d + 1$ subcases.
Note that for nontrivial instances clearly~$d < \ell$.

\begin{algorithm}[t]
  \caption{Improved algorithm for \textsc{ConRMC} (based on Hermelin and Rozenberg~\cite{HR15})}
  \label{algo:wrtl}
  \begin{algorithmic}[1]
    \Input An incomplete matrix \( \mathbf{S} \in (\Sigma \cup \{ * \})^{n \times \ell} \) and \( d_1, \dots, d_n \in \NN \).
    \Output Decide whether there exists a row vector $v \in \Sigma^\ell$ with $d(v, \mathbf{S}[i]) \le d_i$ for all $i \in [n]$. 
    \State \algorithmicif\ $d_i < 0$ for some $i \in [n]$ \algorithmicthen\ \Return \textbf{No}. \label{algo:l:obvno}
    \State \algorithmicif\ $\ell - |P_*(\mathbf{S}[i])| \le d_i$ for all $i \in [n]$ \algorithmicthen\ \Return \textbf{Yes}. \label{algo:l:obvyes}
    \Statex \hfill\algorithmiccomment{$|P_*(\mathbf{S}[i])|$ is the number of missing entries in $\mathbf{S}[i]$}
    \State Choose any $i \in [n]$ such that $\ell - |P_*(\mathbf{S}[i])| > d_i$.
    \State Choose any $R \subseteq [\ell] \setminus P_*(\mathbf{S}[i])$ with $|R| = d_i + 1$. \label{algo:l:subcasetrick}
    \ForAll{$j \in R$} \label{algo:l:subcasetrick2}
      \State Let $\mathbf{S}' = \mathbf{S}[:, [\ell] \setminus \{ j \}]$ and $d_{i'}' = d_i - \delta(\mathbf{S}[i, j], \mathbf{S}[i', j])$ for each $i' \in [n]$.
      \State \algorithmicif\ recursion on $(\mathbf{S}', d_1', \dots, d_n')$ returns \textbf{Yes} \algorithmicthen\ \Return \textbf{Yes}.
    \EndFor
    \State \Return \textbf{No}.
  \end{algorithmic}
\end{algorithm}

\begin{theorem}
  \label{theorem:wrtld}
  For $d := \max_{i \in [n]} d_i$,
  \textsc{ConRMC} can be solved in $O((d + 1)^\ell \cdot n \ell)$ time. 
\end{theorem}
\begin{proof}
  We prove that \Cref{algo:wrtl} is correct by induction on $\ell$.
  More specifically, we show that it returns \textbf{Yes} if there exists a vector $v \in \Sigma^\ell$ that satisfies $\delta(\mathbf{S}[i], v) \le d_i$ for all $i \in [n]$.
  It is easy to see that the algorithm is correct for the base case $\ell = 0$, because it returns \textbf{Yes} if $d_i$ is nonnegative for all $i \in [n]$ and \textbf{No} otherwise (Lines~\ref{algo:l:obvno}~and~\ref{algo:l:obvyes}).
  Consider the case $\ell > 0$.
  The terminating conditions in Lines~\ref{algo:l:obvno}~and~\ref{algo:l:obvyes} are clearly correct.
  We show that branching on~$R$ is correct in Lines~\ref{algo:l:subcasetrick} and~\ref{algo:l:subcasetrick2}.
  If $v[j] \ne \mathbf{S}[i, j]$ holds for all $j \in R$, then we have a contradiction $\delta(v, \mathbf{S}[i]) \ge |R| > d_i$.
  Thus the branching on $R$ leads to a correct output.
  Now the induction hypothesis ensures that the recursion on $\mathbf{S}[:, [\ell] \setminus \{ j \}]$ (notice that it has exactly one column less) returns a desired output.
  This concludes that our algorithm is correct.

  As regards the time complexity, note that each node in the search tree has at most $d + 1$ children.
  Moreover, the depth of the search tree is at most $\ell$ because the number of columns decreases for each recursion.
  Since each recursion step only requires linear (that is, $O(n \ell)$) time, the overall running time is in $O((d + 1)^\ell \cdot n \ell)$.
\end{proof}

Since $d < \ell$ for nontrivial input instances, \Cref{theorem:wrtld} yields 
``linear-time fixed-parameter tractability'' with respect to~$\ell$,
meaning an exponential speedup over the previous result due to Hermelin 
and Rozenberg~\cite{HR15}.

\begin{corollary}
  \label{thm:rmcpfptwrtl}
  \textsc{ConRMC} can be solved in $O(n \ell^{\ell +1})$ time.
\end{corollary}
We remark that this algorithm cannot be significantly improved assuming the ETH.\footnote{The Exponential Time Hypothesis asserts that \textsc{3-SAT} cannot be solved in $O^*(2^{o(n + m)})$ time for a 3-CNF~formula with $n$~variables and $m$~clauses.}
It is known that there is no $\ell^{o(\ell)} \cdot n^{O(1)}$-time algorithm for the special case 
\textsc{Closest String} unless the ETH fails~\cite{LMS18}. 
The running time of our algorithm matches this lower bound (up to a constant in the exponent) and therefore there is presumably no substantially faster algorithm with respect to~$\ell$.

As a consequence of \Cref{thm:rmcpfptwrtl}, we obtain a fixed-parameter algorithm for \textsc{MinLRMC} with respect to the maximum number~$k$ of missing
entries per row in the input matrix.
\begin{corollary}
  \label{thm:minlrmc:fptk}
  \textsc{MinLRMC} can be solved in time $O(n^2 \ell + n^2 k^{k + 1})$.
\end{corollary}

\begin{proof}
  For each $i \in [n]$, we construct an \textsc{ConRMC}-instance, where the input matrix is
$\mathbf{S}_i = \mathbf{S}[:, P_*(\mathbf{S}[i])]$
and
$d_{i, i'} = d - \delta(\mathbf{S}[i], \mathbf{S}[i'])$ for each $i' \in [n]$.
We return \textbf{Yes} if and only if there is a \textbf{Yes}-instance $(\mathbf{S}_i, d_{i, 1}, \dots, d_{i, n})$ of \textsc{ConRMC}.
Each \textsc{ConRMC}-instance requires $O(n \ell)$ time to construct, and $O(n k^{k + 1})$ time to solve, because $\mathbf{S}_i$ contains at most~$k$ columns.
\end{proof}

\section{Combined parameter \texorpdfstring{$d + k$}{d + k}}
\label{sec:radmin:dk}

In this section we generalize two algorithms (one by Gramm et al.~\cite{gramm2003fixed} and one by Ma and Sun~\cite{ma2009more}) for the special case of \textsc{MinRMC} in which the input matrix is complete (known as the \textsc{Closest String} problem) to the case of incomplete matrices.
We will describe both algorithms briefly.
In fact, both algorithms solve the special case of \textsc{ConRMC}, referred to as \textsc{Neighboring String} (generalizing \textsc{Closest String} by 
allowing row-individual distances), where the input matrix is complete.

\dprob{Neighboring String}
{A matrix $\mathbf{T} \in \Sigma^{n \times \ell}$ and $d_1, \dots, d_n \in \NN$}
{Is there a row vector $v \in \Sigma^\ell$ such that $\delta(v, \mathbf{T}[i]) \le d_i$ for each $i \in [n]$?}

\begin{algorithm}[t]
  \caption{Algorithm for \textsc{Neighboring String} by Gramm et al.~\cite{gramm2003fixed}}
  \label{algo:gnr} 
  \begin{algorithmic}[1]
    \Input A matrix \( \mathbf{T} \in (\Sigma \cup \{ * \})^{n \times \ell} \) and \( d_1, \dots, d_n \in \NN \).
    \Output Decide whether there exists a row vector $v \in \Sigma^\ell$ with $d(v, \mathbf{T}[i]) \le d_i$ for all $i \in [n]$.
    \State \algorithmicif\ $d_i < 0$ for some $i \in [n]$ \algorithmicthen\ \Return \textbf{No}. \label{algo:trino}
    \State \algorithmicif\ $\delta(\mathbf{T}[1], \mathbf{T}[i]) \le d_i$ for all $i \in [n]$ \algorithmicthen\ \Return \textbf{Yes}.
    \State Choose any $i \in [n]$ such that $\delta(\mathbf{T}[1], \mathbf{T}[i]) > d_i$.
    \State Choose any $Q' \subseteq Q(\mathbf{T}[1], \mathbf{T}[i])$ with $|Q'| = d_i + 1$.
    \ForAll{$j \in Q'$} \label{algo:branchd}
      \State Let $\mathbf{T}' = \mathbf{T}[:, [\ell] \setminus \{ j \}]$ and $d_{i'}' = d_i - \delta(\mathbf{T}[i, j], \mathbf{T}[i', j])$ for each $i' \in [n]$.
      \State \algorithmicif\ recursion on $(\mathbf{T}', d_1', \dots, d_n')$ returns \textbf{Yes} \algorithmicthen\ \Return \textbf{Yes}.
    \EndFor
    \State \Return \textbf{No}.
  \end{algorithmic}
\end{algorithm}

The algorithm of Gramm et al.~\cite{gramm2003fixed} is given in \Cref{algo:gnr}.
First, it determines whether the first row vector $\mathbf{T}[1]$ is a solution.
If not, then it finds another row vector $\mathbf{T}[i]$ that differs from $\mathbf{T}[1]$ on more than $d_i$ positions and branches on the column positions $Q(\mathbf{T}[1], \mathbf{T}[i])$ where $\mathbf{T}[1]$ and $\mathbf{T}[i]$ disagree.

Using a search variant of \Cref{algo:gnr}, Hermelin and Rozenberg \cite[Theorem 4]{HR15} claimed that \textsc{MinRMC} is fixed-parameter tractable with respect to $d + k$.
Here, we reveal that their algorithm is in fact not correct.
The algorithm chooses an arbitrary row vector $\mathbf{S}[i]$ and calls the algorithm by Gramm et al.~\cite{gramm2003fixed} with input matrix $\mathbf{S}' = \mathbf{S}[:, [\ell] \setminus  P_*(\mathbf{S}[i])]$.
This results in a set of row vectors $v$ satisfying $\delta(v, \mathbf{S}') \le d$.
Then, the algorithm constructs an instance of \textsc{ConRMC} where the input matrix is $\mathbf{S}[P_*(\mathbf{S}[i])]$ and the distance bound is given by $d_{i'} = d - \delta(v, \mathbf{S}'[i'])$ for each $i' \in [n]$.
The correctness proof was based on the erroneous assumption that the algorithm of Gramm et al.~\cite{gramm2003fixed} finds \emph{all} row vectors $v$ satisfying $\delta(v, \mathbf{S}') \le d$ in time $O((d + 1)^d \cdot n \ell)$.
Although Gramm et al.~\cite{gramm2003fixed} noted that this is indeed the case when $d$ is optimal, it is not always true.
In fact, it is generally impossible to enumerate all solutions in time $O((d + 1)^d \cdot n \ell)$ because there can be $\Omega(\ell^d)$ solutions.
We use the following simple matrix to illustrate the error in the algorithm of Hermelin and Rozenberg \cite{HR15}:
\begin{equation*}
  \mathbf{S} = \left[\begin{array}{ccc}
    0 & 1 & 1 \\ 1 & 1 & 1 \\ * & 0 & 0 
  \end{array}\right].
\end{equation*}
We show that the algorithm may output an incorrect answer for $d = 2$.
If the algorithm chooses $i = 3$, then the algorithm by Gramm et al.~\cite{gramm2003fixed} returns only one row vector $00$.
Then the algorithm of Hermelin and Rozenberg \cite{HR15} constructs an instance of \textsc{ConRMC} with $\mathbf{S}' = \begin{bmatrix} 0 & 1 \end{bmatrix}^T$ and $d_1 = d_2 = 0$, resulting in \textbf{No}.
However, the row vector $v = 001$ satisfies $\delta(v, \mathbf{S}) = 2$ and thus the correct output is \textbf{Yes}.
To remedy this, we give a fixed-parameter algorithm for \textsc{MinRMC}, adapting the algorithm by Gramm et al.~\cite{gramm2003fixed}.

Before presenting our algorithm, let us give an observation noted by Gramm et al.~\cite{gramm2003fixed} for the case of no missing entries.
Suppose that the input matrix $\mathbf{S} \in (\Sigma \cup \{ * \})^{n \times \ell}$ contains more than \( nd \) \emph{dirty} columns (a column is said to be dirty if it contains at least two distinct symbols from the alphabet).
Clearly, we can assume that every column is dirty.
For any vector $v \in \Sigma^\ell$, there exists $i \in [n]$ with $\delta(v, \mathbf{S}[i]) \ge d$ by the pigeon hole principle and hence we can immediately conclude that it is a \textbf{No}-instance.
It is easy to see that this argument also holds for \textsc{MinRMC} and thus \textsc{ConRMC}.

\begin{lemma}
  \label{lemma:latmostnd}
  Let $(\mathbf{S}, d_1, \dots, d_n)$ be a \textsc{ConRMC} instance, where $\mathbf{S} \in (\Sigma \cup \{ * \})^{n \times \ell}$ and $d_1, \dots, d_n \in \NN$.
  If \( \mathbf{S} \) contains more than $nd$ dirty columns for $d = \max_{i \in [n]} d_i$, then there is no row vector $v \in \Sigma^\ell$ with $\delta(v, \mathbf{S}[i]) \le d_i$ for all $i \in [n]$.
\end{lemma}

\begin{algorithm}[t]
  \caption{Algorithm for \textsc{ConRMC} (generalizing \Cref{algo:gnr})}
  \label{algo:gnr2}
  \begin{algorithmic}[1]
  \Input An incomplete matrix \( \mathbf{S} \in (\Sigma \cup \{ * \})^{n \times \ell} \) and \( d_1, \dots, d_n \in \NN \).
  \Output Decide whether there exists a row vector $v \in \Sigma^\ell$ with $d(v, \mathbf{S}[i]) \le d_i$ for all $i \in [n]$.
  \If {$d_1 = 0$} \label{algo:d1zero}
    \State Let $\mathbf{S}' = \mathbf{S}[[2, n], P_*(\mathbf{S}[1])]$ and $d_i' = d_i - \delta(\mathbf{S}[1], \mathbf{S}[i])$ for each $i \in [2, n]$.
    \State \Return the output of \Cref{algo:wrtl} on $(\mathbf{S}', d_2', \dots, d_n')$. \label{algo:d1zeroout}
  \EndIf
  \State Let $R_i = (P_*(\mathbf{S}[1]) \setminus P_*(\mathbf{S}[i])) \cup Q(\mathbf{S}[1], \mathbf{S}[i])$ for each $i \in [2, n]$.
  \State \algorithmicif\ \( |R_i| \le d_i \) for all $i \in [2, n]$ \algorithmicthen\ \Return \textbf{Yes}. \label{algo:obviousyes}
  \State Choose any $i \in [n]$ with $|R_i| > d_i$.
  \State Choose any $R \subseteq R_i$ with $|R| = d_i + 1$. \label{algo:choiceofqprime}
  \ForAll{$j \in R$} \label{algo:jbranch}
    \State Let $\mathbf{S'} = \mathbf{S}[:, [\ell] \setminus \{ j \}]$ and $d_{i'}' = d_i - \delta(\mathbf{S}[i, j], \mathbf{S}[i', j])$ for each $i' \in [n]$.
    \State \algorithmicif\ recursion on $(\mathbf{S}', d_1', \dots, d_n')$ returns \textbf{Yes} \algorithmicthen\ \Return \textbf{Yes}. \label{algo:jrecursion}
  \EndFor
  \State \Return \textbf{No}.
  \end{algorithmic}
\end{algorithm}

Our algorithm is given in \Cref{algo:gnr2}.
It generalizes \Cref{algo:gnr} and finds the solution vector even if the input matrix is incomplete.
In contrast to \textsc{Neighboring String}, the output cannot be immediately determined even if $d_1 = 0$.
We use \Cref{algo:wrtl} to overcome this issue (Line~\ref{algo:d1zeroout}).
\Cref{algo:gnr2} also considers the columns where the first row vector has missing entries (recall that $P_*(\mathbf{S}[1])$ denotes column indices $j$ with $\mathbf{S}[1, j] = *$) in the branching step (Line~\ref{algo:jbranch}), and not only the columns where $\mathbf{S}[1]$ and $\mathbf{S}[i]$ disagree.
Again, we restrict the branching to $d_i + 1$ subcases (Line~\ref{algo:choiceofqprime}).
This reduces the size of the search tree significantly.
We show the correctness of \Cref{algo:gnr2} and analyze its running time in the proof of the following theorem.

\begin{theorem}
  \label{theorem:fpt}
  For $d = \max_{i \in [n]} d_i$, \textsc{ConRMC} can be solved in $O(n \ell + (d + 1)^{d + k + 1} n)$ time.
\end{theorem}
\begin{proof}
  First, we prove that \Cref{algo:gnr2} is correct by induction on $d_1 + |P_*(\mathbf{S}[1])|$.
  More specifically, we show that the algorithm returns \textbf{Yes} if and only if a vector $v \in \Sigma^\ell$ satisfying $\delta(\mathbf{S}[i], v) \le d_i$ for all $i \in [n]$ exists.

  Consider the base case $d_1 + |P_*(\mathbf{S}[1])| = 0$. 
  Since $d_1 = 0$, the algorithm terminates in Line~\ref{algo:d1zeroout}.
  When $d_1 = 0$, any solution vector must agree with $\mathbf{S}[1]$ on each entry unless the entry is missing in $\mathbf{S}[1]$.
  Hence, the output in Line~\ref{algo:d1zeroout} is correct by \Cref{theorem:wrtld}.
  Consider the case $d_1 + |P_*(\mathbf{S}[1])| > 0$.
  Let $R_i = (P_*(\mathbf{S}[1]) \setminus P_*(\mathbf{S}[i])) \cup Q(\mathbf{S}[1], \mathbf{S}[i])$ for each $i \in [2, n]$. 
  If $|R_i| \le d_i$ holds for all $i \in [2, n]$, then the vector $\mathbf{S}[1] \oplus \sigma^\ell$ (the vector obtained by filling each missing entry in $\mathbf{S}[1]$ with $\sigma$) is a solution for an arbitrary character $\sigma \in \Sigma$.
  Hence, Line~\ref{algo:obviousyes} is correct.
  Suppose that there exists a solution vector $v \in \Sigma^\ell$ with $\delta(v, \mathbf{S}[i]) \le d_i$ for all $i \in [n]$.
  We show that the branching in Line~\ref{algo:jbranch} is correct.
  Let $R$ be as specified in Line~\ref{algo:choiceofqprime}.
  We claim that there exists a $j \in R$ with $v[j] = \mathbf{S}[i, j]$ for every choice of $R$.
  Otherwise, $v[j] \ne \mathbf{S}[i, j]$ and $\mathbf{S}[i, j] \ne *$ holds for all $j \in R$ and we have $\delta(v, \mathbf{S}[i]) > d_i$ (a contradiction).
  Note that $\mathbf{S}[:, [\ell] \setminus \{ j \}]$ has exactly one less missing entry if $j \in P_*(\mathbf{S}[1])$ and that $d_1' = d_1 - 1$ in case of $j \in Q(\mathbf{S}[1], \mathbf{S}[i])$.
  It follows that $d_1 + |P_*(\mathbf{S}[1])|$ is strictly smaller in the recursive call (Line~\ref{algo:jrecursion}).
  Hence, the induction hypothesis ensures that the algorithm returns \textbf{Yes} when $v[j] = \mathbf{S}[i, j]$ holds.
  On the contrary, it is not hard to see that the algorithm returns \textbf{No} if there is no solution vector.
  Thus, \Cref{algo:gnr2} is correct.

  We examine the time complexity.
  Assume without loss of generality that $k = |P_*(\mathbf{S}[1])|$ and $d = d_1$ hold initially.
  Consider the search tree where each node corresponds to a call on either \Cref{algo:wrtl} or \Cref{algo:gnr2}.
  If $d_1 > 0$, then $d_1 + P_*(\mathbf{S}[1])$ decreases by 1 in each recursion and there are at most $d + 1$ recursive calls.
  Let $u$ be some node in the search tree that invokes \Cref{algo:wrtl} for the first time.
	We have seen in the proof of \Cref{theorem:wrtld} that the subtree rooted at $u$ is a tree of depth at most $|P_*(\mathbf{S}[1])|$, in which each node has at most $d_i - \delta(\mathbf{S}[1], \mathbf{S}[i]) + 1 \le d + 1$ children.
  Note also that $u$ lies at depth $d + k - |P_*(\mathbf{S}[1])|$.
  Thus, the depth of the search tree is at most $d + k$ and the search tree has size $O((d + 1)^{d + k})$.
  We can assume that $\ell \le nd$ by \Cref{lemma:latmostnd} and hence each node requires $O(nd)$ time.
  This results in the overall running time of $O(n \ell + (d + 1)^{d + k + 1} n)$.
\end{proof}

\begin{algorithm}[t]
  \caption{Algorithm for \textsc{Neighboring String} by Ma and Sun \cite{ma2009more}}
  \label{algo:ms}
  \begin{algorithmic}[1]
    \Input A matrix \( \mathbf{T} \in (\Sigma \cup \{ * \})^{n \times \ell} \) and \( d_1, \dots, d_n \in \NN \).
    \Output Decide whether there exists a row vector $v \in \Sigma^\ell$ with $d(v, \mathbf{T}[i]) \le d_i$ for all $i \in [n]$.
    \State \algorithmicif\ $\delta(\mathbf{T}[1], \mathbf{T}[i]) \le d_i$ for all $i \in [n]$ \algorithmicthen\ \Return \textbf{Yes} \label{algo:terminal2}.
    \State Choose any $i \in [n]$ such that $\delta(\mathbf{T}[1], \mathbf{T}[i]) > d_i$.
    \State Let $Q = Q(\mathbf{T}[1], \mathbf{T}[i])$.
    \ForAll{$v \in \Sigma^{|Q|}$ such that $\delta(v, \mathbf{T}[1]) \le d_1$ and $\delta(v, \mathbf{T}[i]) \le d_i$} \label{algo:jbranchall}
      \State Let $\mathbf{T}' = \mathbf{T}[:, [\ell] \setminus Q]$ and $d_1' = \min \{ d_1 - \delta(v, \mathbf{T}[1, Q]), \lceil d_1 / 2 \rceil - 1\}$. \label{algo:d1half}
      \State Let $d_{i'}' = d_{i'} - d(v, \mathbf{T}[i', Q])$ for each $i' \in [2, n]$.
      \State \algorithmicif\ recursion on $(\mathbf{T}, d_1', \dots, d_n')$ returns \textbf{Yes} \algorithmicthen\ \Return \textbf{Yes}.
    \EndFor
    \State \Return \textbf{No}.
  \end{algorithmic}
\end{algorithm}

Now, we provide a more efficient fixed-parameter algorithm when the alphabet size is small, based on \Cref{algo:ms} by Ma and Sun \cite{ma2009more}.
Whereas \Cref{algo:gnr} considers each position of (a subset of) $Q(\mathbf{T}[1], \mathbf{T}[i])$ one by one, \Cref{algo:ms} considers all vectors on $Q(\mathbf{T}[1], \mathbf{T}[i])$ in a single recursion. 
The following lemma justifies why $d_1$ can be halved (Line~\ref{algo:d1half}) in each iteration (the vectors $u$ and $w$ correspond to $\mathbf{T}[1]$ and $\mathbf{T}[i]$, respectively).

\begin{lemma}{\cite[Lemma 3.1]{ma2009more}}
  \label{prop:drelation}
  Let $u, v, w \in \Sigma^\ell$ be row vectors satisfying $\delta(u, w) > \delta(v, w)$.
  Then, it holds that $\delta(u[Q'], v[Q']) < \delta(u, v) / 2$ for $Q' = [\ell] \setminus Q(u, w)$.
\end{lemma}
\begin{proof}
  Assume that $\delta(u[Q'], v[Q']) \ge \delta(u, v) / 2$.
  We can rewrite the value of $\delta(u, v) + \delta(v, w)$ as follows:
  \begin{align*}
    \delta(u, v) + \delta(v, w)
    &= \delta(u[Q'], v[Q']) + \delta(v[Q'], w[Q']) + \delta(u[Q], v[Q]) + \delta(v[Q], w[Q]),
  \end{align*}
  where $Q$ is a shorthand for $Q = Q(v, w)$.
  It follows from the definition of $Q'$ that $u[Q'] = w[Q']$ and hence
  \begin{align}
    \label{eq:breakww}
    \delta(u[Q'], v[Q']) = \delta(v[Q'], w[Q']).
  \end{align}
  We also note that $\delta(u[Q] + v[Q]) + \delta(v[Q] + w[Q]) \ge |Q| = \delta(v, w)$ because it must hold that $u[j] \ne v[j]$ or $v[j] \ne w[j]$ for each $j \in Q$.
  Now, we obtain the following contradiction concluding the proof:
  \begin{align*}
    \delta(u, v) + \delta(v, w)
    \ge 2 \delta(u[Q'], v[Q']) + \delta(u, w)
    > \delta(u, v) + \delta(v, w).
  \end{align*}
\end{proof}

\Cref{prop:drelation} plays a crucial role in obtaining the running time $O(n\ell + (16|\Sigma|)^d nd)$ of Ma and Sun \cite{ma2009more}.
However, \Cref{prop:drelation} may not hold in the presence of missing entries (in fact, \Cref{eq:breakww} may break when at least one of $u$ or $w$ contains missing entries).
For instance, $u = 0^\ell, v = 1^\ell, w = *^\ell$ is one counterexample.
Note that here $Q(u, w) = \emptyset$ and that $\delta(u[Q'], v[Q']) = \delta(u, v) = \ell$.
To work around this issue, let us introduce a new variant of \textsc{Closest String}
which will be useful to derive a fixed-parameter algorithm for \textsc{ConRMC} (\Cref{theorem:fpt2}).
We will use a special character ``$\diamond$'' to denote a ``dummy'' character.

\dprob{Neighboring String with Dummies (NSD)}
{A matrix $\mathbf{T} \in (\Sigma \cup \{ \diamond \})^{n \times \ell}$ and $d_1, \dots, d_n \in \NN$.}
{Is there a row vector $v \in \Sigma^\ell$ such that $\delta(v, \mathbf{T}[i]) \le d_i$ for each $i \in [n]$?}

Note that the definition of \textsc{NSD} forbids dummy characters in the solution vector~$v$.
Observe that \Cref{prop:drelation} (in particular \Cref{eq:breakww}) holds even if row vectors $u$ and $w$ contain dummy characters.
We show that \textsc{NSD} can be solved using \Cref{algo:ms} as a subroutine.

\begin{lemma} \label{lem:NSD}
\textsc{NSD} can be solved in $O(n \ell + |\Sigma|^{k} \cdot nk + 2^{4d - 3k} \cdot |\Sigma|^{d} \cdot nd))$ time, where $d := \max_{i \in [n]} d_i$ and $k$ is the minimum number of dummy characters in any row vector of $\mathbf{T}$.
\end{lemma}
\begin{proof}
  With \Cref{prop:drelation}, assuming $d=d_1$ one can prove by induction on~$d_1$ that \Cref{algo:ms} solves the \textsc{NSD} problem if the first row vector $\mathbf{T}[1]$ contains no dummy characters by induction on $d_1$.
  Refer to \cite[Theorem~3.2]{ma2009more} for details.
  We describe how we use \Cref{algo:ms} of Ma and Sun \cite{ma2009more} to solve \textsc{NSD}.
  Let $I = (\mathbf{T}, d_1, \dots, d_n)$ be an instance of \textsc{NSD}.
  We assume that $|P_\diamond(\mathbf{T}[1])| = k$.
  For each row vector $u$ of $\Sigma^{k}$, we invoke \Cref{algo:ms} with the input matrix $\mathbf{T}' = \mathbf{T}[[\ell] \setminus P_\diamond(\mathbf{T}[1])]$ and the distance bounds $d_1 - k, d_2 - \delta(u, \mathbf{T}[2, P_\diamond(\mathbf{T}[1])]), \dots, d_n - \delta(u, \mathbf{T}[n, P_\diamond(\mathbf{T}[1])])$.
  Note that $\mathbf{T}'[1]$ contains no dummy character and thus the output of \Cref{algo:ms} is correct.
  We return \textbf{Yes} if and only if \Cref{algo:ms} returns \textbf{Yes} at least once.
  Let us prove that this solves \textsc{NSD}.
  If $I$ is a \textbf{Yes}-instance with solution vector $v \in \Sigma^\ell$, then it is easy to verify that \Cref{algo:ms} returns \textbf{Yes} when \( u = v[P_\diamond(\mathbf{T}[1])] \).
  On the contrary, the distance bounds in the above procedure ensure that $I$ is a \textbf{Yes}-instance if \cref{algo:ms} returns \textbf{Yes}.

  Now we show that this procedure runs in the claimed time.
  Ma and Sun \cite{ma2009more} proved that \Cref{algo:ms} runs in 
  \begin{align*}
    O \left(n \ell + \binom{d_{\max} + d_{\min}}{d_{\min}} \cdot (4|\Sigma|)^{d_{\min}} \cdot n d_{\max} \right)
  \end{align*}
  time, where $d_{\max} = \max_{i \in [n]} d_i$ and $d_{\min} = \min_{i \in [n]} d_i$.
  In fact, they showed that each node in the search tree requires $O(n d_{\max})$ time by remembering previous distances, as it only concerns $O(d_{\max})$ columns.
  In the same spirit, one can compute distances from the first row vector for each \textsc{NSD}-instance under consideration in $O(nk)$ time, given the corresponding distances in the input matrix.
  Since we have \( d_{\max} \le d \) and \( d_{\min} \le d - k \) for each call of \Cref{algo:ms}, it remains to show that \( \binom{2d - k}{d} \in O(2^{2d - k}) \).
  Using Stirling's approximation \( \sqrt{2 \pi} n^{n + 1/2}e^{-n} \le n! \le e n^{n + 1/2}e^{-n} \) which holds for all positive integers~$n$, we obtain
  \begin{align*}
    \binom{2d - k}{d} = \frac{(2d - k)!}{d! \cdot (d - k)!} \le c \cdot \frac{(2d - k)^{2d - k}}{d^d \cdot (d - k)^{d - k}}
  \end{align*}
  for some constant $c$.
  We claim that the last term is upper-bounded by $c \cdot 2^{2d - k}$.
  We use the fact that the function $x \mapsto x \log x$ is convex over its domain $x > 0$ (note that the second derivative is given by \( x \mapsto 1 / x \)).
  Since a convex function $f \colon D \to \RR$ satisfies $(f(x) + f(y)) / 2 \ge f((x + y) / 2)$ for any $x, y \in D$, we obtain
  \begin{align*}
    d \log d + (d - k) \log (d - k) \ge 2 \left( \frac{2d - k}{2} \right) \cdot \log \left( \frac{2d - k}{2} \right).
  \end{align*}
  It follows that \( d^d \cdot (d - k)^{d - k} = 2^{d \log d + (d - k) \log (d - k)} \ge 2^{(2d - k) \log(d - k / 2)} = (d - k / 2)^{2d - k }\).
  This shows that $\binom{2d - k}{d} \in O(2^{2d - k})$.
\end{proof}

Finally, to show our second main result in this section, we provide a polynomial-time reduction from \textsc{ConRMC} to \textsc{NSD}.

\begin{theorem}
  \label{theorem:fpt2}
\textsc{ConRMC} can be solved in $O(n \ell + 2^{4d + k} \cdot |\Sigma|^{d + k} \cdot n (d + k))$ time.
\end{theorem}
\begin{proof}
  Let $I = (\mathbf{S}, d_1, \dots, d_n)$ be an instance of \textsc{ConRMC}.
  We construct an instance $I' = (\mathbf{T}, d_1 + k, \dots, d_n + k)$ of \textsc{NSD} where $\mathbf{T} \in (\Sigma \cup \{ \diamond \})^{n \times (\ell + k)}$ and each row vector of \( \mathbf{T} \) contains exactly $k$ dummy characters.
  Note that such a construction yields an algorithm running in $O(n (\ell + k) + |\Sigma|^{k} \cdot nk + 2^{4d + k} \cdot |\Sigma|^{d + k} \cdot n (d + k)) = O(n \ell + 2^{4d + k} \cdot |\Sigma|^{d + k} \cdot n (d + k))$ time using \Cref{lem:NSD}.
  Let $\sigma \in \Sigma$ be an arbitrary character.
  We define the row vector $\mathbf{T}[i]$ for each $i\in [n]$ as follows:
  Let $\mathbf{T}[i, [\ell]] = \mathbf{S}[i] \oplus \diamond^\ell$ (in other words, the row vector $\mathbf{T}[i, [\ell]]$ is obtained from $\mathbf{S}[i]$ by replacing $*$ by $\diamond$) for the leading $\ell$ entries.
  For the remainder, let
  \begin{align*}
    \mathbf{T}[i, \ell + j] =
    \begin{cases}
      \sigma & \text{if \(j \le |P_*(\mathbf{S}[i])|\)}, \\
      \diamond & \text{otherwise},
    \end{cases}
  \end{align*}
  for each $j \in [k]$.
  See \Cref{fig:dummyreduction} for an illustration.
  We claim that $I$ is a \textbf{Yes} instance if and only if $I'$ is a \textbf{Yes} instance.

  $(\Rightarrow)$
  Let $v \in \Sigma^\ell$ be a solution of $I$.
  We claim that the vector $v' \in \Sigma^{\ell + k}$ with $v'[[\ell]] = v$ and $v[[\ell + 1, \ell + k]] = \sigma^k$ is a solution of $I'$.
  For each $i \in [n]$, we have
  \[
    \delta(v', \mathbf{T}[i])
    = \delta(v'[[\ell]], \mathbf{T}[i, [\ell]]) + \delta(\sigma^k, \mathbf{T}[i, [\ell + 1, \ell + k]]).
  \]
  It is easy to see that the first term is at most \(d_i + |P_*(\mathbf{S}[i])|\) and that the second term equals $k - |P_*(\mathbf{S}[i])|$.
  Thus we have $\delta(v', \mathbf{T}) \le d_i + k$.

  $(\Leftarrow)$
  Let $v' \in \Sigma^\ell$ be a solution of $I'$.
  Since the row vector $\mathbf{T}[i, [\ell + 1, \ell + k]]$ contains $k - |P_*(\mathbf{S}[i])|$ dummy characters, we have \( \delta(v'[[\ell]], \mathbf{T}[i, [\ell]]) \le (d_i + k) - (k - |P_*(\mathbf{S}[i])|) = d_i + |P_*(\mathbf{S}[i])| \) for each $i \in [n]$.
  It follows that $\delta(v'[[\ell]], \mathbf{S}[i]) \le d_i$ holds for each $i \in [n]$.
\end{proof}

Note that the algorithm of \Cref{theorem:fpt2} is faster than the algorithm of~\Cref{theorem:fpt} for $|\Sigma| < d / 16$ and faster than the $O^*(|\Sigma|^k \cdot d^d)$-time algorithm by Hermelin and Rozenberg~\cite{HR15} for $|\Sigma| < d / 2^{4 + d/k}$.

\begin{figure}[t]
  \centering
  \begin{subfigure}[t]{.3\textwidth}
    \begin{align*}
      \mathbf{S} = \left[\begin{array}{ccc}
        0 & 0 & 0 \\
        1 & 1 & * \\
        * & * & 2 
      \end{array}\right]
    \end{align*}
  \end{subfigure}
  \begin{subfigure}[t]{.4\textwidth}
    \begin{align*}
      \mathbf{T} = \left[\begin{array}{*{5}{c}}
        0 & 0 & 0 & \diamond & \diamond \\
        1 & 1 & \diamond & \sigma   & \diamond \\
        \diamond & \diamond & 2 & \sigma   & \sigma 
      \end{array}\right]
    \end{align*}
  \end{subfigure}
  \caption{An illustration of the reduction in \Cref{theorem:fpt2}.
  Given the matrix $\mathbf{S}$ with \( k = 2 \) (left), our reduction constructs the matrix \( \mathbf{T} \) with \( k \) additional columns (right).
  Note that every row vector in \( \mathbf{T} \) contains exactly two dummy characters.
  The \textsc{MinRMC} instance $(\mathbf{S}, d = 1)$ is a \textbf{Yes}-instance with a solution vector $v = 100$.
  The corresponding \textsc{NSD}-instance $(\mathbf{T}, d + k = 3)$ is also a \textbf{Yes}-instance with a solution vector $v' = 100\sigma\sigma$.}
  \label{fig:dummyreduction}
\end{figure}

\section{Conclusion}
\label{sec:minrmc:conc}
We studied problems appearing both in stringology and in the context of matrix 
completion. The goal in both settings 
is to find a 
consensus string (matrix row) that is close to all given input strings (rows). 
The special feature here now is the existence of wildcard letters 
(missing entries) appearing in the 
strings (rows). Thus, these problems naturally generalize the 
well-studied \textsc{Closest String} and related string problems.
Although with applications in the context of data mining, machine learning, and
computational biology
at least as well motivated as \textsc{Closest String},
so far there is comparatively little work on these ``wildcard problems''. 
This work is also meant to initiate further research in this direction.

We conclude with a list of challenges for future research:
\begin{itemize}
  \item
    Can the running time of \Cref{theorem:fpt2} be improved?
    Since Ma and Sun \cite{ma2009more} proved that \textsc{Closest String} can be solved in $O((16|\Sigma|)^d \cdot n \ell)$ time, a plethora of efforts have been made to reduce the base in the exponential dependence in the running time \cite{CW11,NS12,CMW12,CMW16}.
		A natural question is whether these results can be translated to \textsc{MinRMC} and \textsc{ConRMC} as well.
  \item
    Another direction would be to consider the generalization of \textsc{MinRMC} with \emph{outliers}.
    The task is to determine whether there is a set $I \subseteq [n]$ of row indices and a vector $v \in \Sigma^\ell$ such that $|I| \le t$ and $\delta(v, \mathbf{S}[[n] \setminus I]) \le d$.
    For complete input matrices, this problem is known as \textsc{Closest String with Outliers} and
    fixed-parameter tractability with respect to $d + t$ is known~\cite{BM11}. 
    Hence, it is interesting to study whether the outlier variant of \textsc{MinRMC} (or \textsc{ConRMC}) is fixed-parameter tractable with respect to $d + k + t$.
  \item
    Finally, let us mention a maximization variant \textsc{MaxRMC} where the goal is to have a radius at least~$d$.
    The complete case is referred to as \textsc{Farthest String} \cite{WZ09} and fixed-tractability with respect to $|\Sigma| + d$ is known \cite{GGN06,WZ09}.
    Is \textsc{MaxRMC} also fixed-parameter tractable with respect to $(d,|\Sigma|)$?
   
\end{itemize}

\appendix


\bibliography{ref}

\end{document}